\documentclass[conference]{IEEEtran}

\usepackage[cmex10]{amsmath}
\usepackage{amsfonts}
\usepackage{amssymb}
\usepackage{bbm}
\usepackage{verbatim} 
\usepackage{graphics,psfrag,epsfig}
\usepackage{cases}
\usepackage[linesnumbered,ruled,vlined,algo2e]{algorithm2e}
\usepackage{empheq}

\usepackage{amsthm} 
\usepackage{hyperref}

\usepackage{color}
\definecolor{red}{rgb}{1,0.2,0.2}
\definecolor{green}{rgb}{0.2,1,0.5}
\definecolor{blue}{rgb}{0,0,1}
\definecolor{lightblue}{rgb}{0.3,0.5,1}

\newcommand{\diag}{\mathtt{diag}}

\newcommand{\0}{\mathbf{0}}
\newcommand{\1}{\mathbf{1}}

\newcommand{\cb}{\mathbf{c}}

\newcommand{\bsigma}{\boldsymbol{\sigma}}

\newcommand{\x}{{\mathbf{x}}}

\newcommand{\RN}{\mathbb{R}^N}

\newcommand{\Rnn}{\mathbb{R}^{n\times n}}

\newcommand{\gb}{\mathbf{g}} 

\newcommand{\pb}{\mathbf{p}}
\newcommand{\qb}{\mathbf{q}}
\newcommand{\bs}{\mathbf{s}}

\newcommand{\T}{^{\mathsf{T}}}

\newcommand{\tb}{\mathbf{t}}

\newcommand{\xx}{{\mathbf{x}}}
\newcommand{\yb}{\mathbf{y}}

\newcommand{\nnb}{\nonumber}

\newtheorem{thm}{Theorem}%
\newtheorem{lem}{Lemma}
\newtheorem{prop}{Proposition}
\newtheorem{cor}{Corollary}

\theoremstyle{remark}

\newcommand {\cA}{{\mathcal{A}}}

\newcommand {\cG}{{\mathcal{G}}}

\newcommand {\cS}{{\mathcal{S}}}

\newcommand {\cV}{{\mathcal{V}}}

\newcommand {\bd} {{\bf d}}

\newcommand {\bh} {{\bf h}}

\newcommand {\bp} {{\bf p}}
\newcommand {\bq} {{\bf q}}

\newcommand {\bx} {{\bf x}}
\newcommand {\by} {{\bf y}}

\newcommand {\blambda} {\boldsymbol{\lambda}}
\newcommand {\bdelta} {\boldsymbol{\delta}}
\newcommand {\bvarepsilon} {\boldsymbol{\varepsilon}}
\newcommand {\bmu} {\boldsymbol{\mu}}

\newcommand {\N} {{\mathbb{N}}}
\newcommand {\R} {\mathbb{R}}
\newcommand {\C} {{\rm I\kern-5pt C}}

\newcommand{\beqa}{\begin{eqnarray}}
\newcommand{\eeqa}{\end{eqnarray}}
\newcommand{\beqan}{\begin{eqnarray*}}
	\newcommand{\eeqan}{\end{eqnarray*}}
\newcommand{\beq}{\begin{equation}}
\newcommand{\eeq}{\end{equation}}

\newcommand{\bfl}{\begin{flushleft}}
	\newcommand{\efl}{\end{flushleft}}
\newcommand{\bgnv}{\begin{verbatim}}
	\newcommand{\ednv}{\end{verbatim}}

\newcommand{\myb}{\hspace{-0.1in}}

\newcommand{\myeq}{& \hspace{-0.1in} = & \hspace{-0.1in}}

\newcommand{\lb}{\nonumber \\}
\newcommand{\myarr}{\begin{array}{lll}}

	\newcommand{\sC}{^{(C)}}
	\newcommand{\mygeq}{& \myb \geq & \myb}

	\newcommand{\bitem}{\begin{itemize}}
		\newcommand{\eitem}{\end{itemize}}
	\newcommand{\benum}{\begin{enumerate}}
		\newcommand{\eenum}{\end{enumerate}}

	\newcommand{\myhb}{\hspace{-0.3in}}

	\newcommand{\myskip}{\\ \vspace{-0.1in}}
	\newcommand{\corres}{%
  \mathrel{%
    \vcenter{\offinterlineskip
      \ialign{##\cr$\to$\cr\noalign{\kern-1.5pt}$\to$\cr}%
    }%
  }%
}

\begin{document}
		
\title{Optimal Cybersecurity Investments Using SIS Model: Weakly Connected Networks}
		
\author{\IEEEauthorblockN{Van Sy Mai}
\IEEEauthorblockA{\textit{CTL, NIST} \\
Gaithersburg, MD USA \\
vansy.mai@nist.gov}
\and
\IEEEauthorblockN{Richard J. La}
\IEEEauthorblockA{\textit{ACMD, NIST 
\& Univ. of MD} \\
College Park, MD USA \\
hyongla@umd.edu}
\and
\IEEEauthorblockN{Abdella Battou}
\IEEEauthorblockA{\textit{CTL, NIST} \\
Gaithersburg, MD USA \\
abdella.battou@nist.gov}
}

\maketitle
		
\begin{abstract}
We study the problem of minimizing the (time) 
average security costs in large systems 
comprising many interdependent subsystems, 
where the state evolution is captured by a susceptible-infected-susceptible (SIS) model. 
The security costs reflect 
security investments, economic losses
and recovery costs from infections and failures 
following successful attacks. However, unlike in 
existing studies, we assume that the underlying
dependence graph is only weakly connected, but
not strongly connected. When the dependence graph
is not strongly connected, existing approaches to 
computing optimal security investments cannot be
applied. Instead, we show that it is still possible
to find a good solution by {\em perturbing} the
problem and establishing necessary continuity results
that then allow us to leverage the existing algorithms.
\end{abstract}

\section{Introduction}

In complex engineering systems, comprising systems work together to deliver their services, e.g., information and communication networks and power systems, introducing interdependence among the systems. This interdependence among systems also allows a local failure and infection of a system by malware to spread to other systems. Similarly, infectious diseases propagate via contacts in social networks. Therefore, the structure of interdependence among the systems should be taken into consideration when determining their security investments.

There is already a large volume of literature that examines how to optimize the (security) investments in complex systems or the mitigation of disease spread. For example, in \cite{Khalili-NetCon19, La-ToN16, LelargeBolot-NetEcon08, Hota-TCNS18}, researchers adopted a game theoretic formulation to study the problem of security investments with distributed agents. In another line of research, which is more closely related to our study, researchers investigated optimal strategies using vaccines/immunization (prevention) \cite{Cohen-PRL03,Preciado-CDC13}, antidotes or curing rates (recovery) \cite{Borgs-RSA10,  mai2018distributed, Ottaviano-JCN18} or a combination of both preventive and recovery measures \cite{Nowzari-TCNS17, Preciado-TCNS14}. In addition, in recent studies~\cite{mai_globecom20, mai_ToN21}, Mai et al. investigated the problem of minimizing the (time) average costs of a system operator, where the costs includes both security investments and recovery/repair costs ensuing infections or failures. 

All of theses studies assume that the underlying dependence graph is strongly connected. In this study, we adopt the framework used in \cite{mai_globecom20, mai_ToN21} but allow the dependence graph to be {\em weakly} connected. When the graph is only weakly connected, some of key properties and results proved for strongly connected networks do not hold. As a result, we cannot directly apply the algorithms from existing studies, including those of \cite{mai_ToN21}. 

In a related study, Khanafer et al. \cite{Khanafer-CDC14} extended the earlier studies on the stability of the susceptible-infectious-susceptible (SIS) model (e.g., \cite{Preciado-CDC13}) to weakly connected networks. A weakly connected network comprises a set of strongly connected components (SCCs) $\{ S_1, S_2, \ldots, S_n \}$. They assumed that the $n$ SCCs can be ordered $S_1 \prec S_2 \prec \cdots \prec S_n$, where $S_i \prec S_{i+1}$ indicates the presence of a directed path from a node in $S_i$ to another node in $S_{i+1}$ but not vice versa, and proved the following: (i) if every SCC has a reproduction number less than 1, then the disease-free state is the unique globally asymptotically stable (GAS) equilibrium; and (ii) if $S_1$ has a reproduction number larger than 1 and every other SCC has a reproduction number smaller than 1, then there is a unique endemic GAS equilibrium. 

In our model, attacks targeting systems arrive according to some (stochastic) process. Successful infections of systems can spread to other systems via dependence among the systems. The system operator decides suitable security investments to fend off the attacks, which in turn determine the {\em breach probability} that they fall victim to attacks and become infected. Our goal is to minimize the (time) average costs of the system operator managing a large system comprising many systems, e.g., large enterprise intranets. The overall costs in our model account for both security investments and recovery/repair costs ensuing infections, which we call {\em infection costs}.
\myskip

{\bf {\em Contributions:}} This paper presents an important extension of the work reported in \cite{mai_ToN21} to more general and common situations in practice, where the underlying dependence graph is only weakly connected. Our approach based on perturbation of the external attack rates allows us to leverage efficient methods for solving the nonconvex perturbed problem approximately. In particular, we show that the optimal point and optimal value of the problem are continuous and increasing in the perturbation vector. As a result, we can solve the perturbed problem instead, for which suboptimality can be quantified using computable upper and lower bounds on the optimal value. We also provide a sufficient condition under which these bounds coincide, i.e., the perturbed problem can be solved exactly despite its nonconvexity. 
\myskip

The rest of the paper is organized as follows:
Section~\ref{sec:Prelim} explains the notation 
and terminology we adopt. Section~\ref{sec:Formulation}
describes the setup and the problem formulation. 
Section~\ref{sec:Perturbed} presents the perturbed problem and the main results, followed by our solution approach in Section~\ref{sec:Proposed}. 
Section~\ref{sec:Numerical} provides some numerical results and 
Section~\ref{sec:Conclusion} concludes the paper.

\section{Preliminaries}
		\label{sec:Prelim}
		
		\subsection{Notation and Terminology}
		\label{subsec:Notation}
		
		Let $\mathbb{R}$ and $\mathbb{R}_+$ denote the 
		set of real numbers and nonnegative real 
		numbers, respectively. 
		%
		For a matrix $A=[a_{i,j}]$, let $a_{i,j}$ denote 
		its $(i,j)$ element, $A\T$ its transpose, 
		and $\rho(A)$ its spectral radius. 
		For two matrices $A$ and $B$, we write $A \!\ge\! B$ if 
		$A\!-\!B$ is a nonnegative matrix. 
		%
		We use boldface letters and numbers to denote vectors, e.g., 
		$\x \!=\! [x_1,..., x_n]\T$ and 
		$\1 \!=\! [1,...,1]\T$. 
		For any two vectors $\xx$ and $\yb$ of the same 
		dimension, $\xx \circ \yb$ and $\frac{\xx}{\yb}$ are 
		their element-wise product and division, respectively. 
		For $\xx \in \mathbb{R}^n$, $\diag(\xx) \in  
		\mathbb{R}^{n \times n}$ denotes the diagonal matrix 
		with diagonal elements $x_1,\ldots,x_n$.
		
		A directed graph $\mathcal{G} \!=\! (\mathcal{V}, 
		\mathcal{E})$ consists of a set of nodes $\mathcal{V}$ 
		and a set of directed edges $\mathcal{E}  \subseteq  
		\mathcal{V} \times \mathcal{V}$. A directed path is a 
		sequence of edges in the form $\big( (i_1, i_2), 
		(i_2, i_3),..., (i_{k-1}, i_k) \big)$. The graph 
		$\mathcal{G}$ is strongly connected if there is a 
		directed path from each node to any other node. The directed graph $\mathcal{G}$ is said to be weakly connected if the undirected graph we obtain after replacing its directed edges with undirected edges is connected.
		\\ \vspace{-0.1in}

		\subsection{M-Matrix Theory} 
		\label{subsec:M-matrix}
		
		A matrix $A \in \R^{n \times n}$ is a Z-matrix if the off-diagonal elements of $A$ are 
		nonpositive and the diagonal elements are nonnegative. 
		A matrix $A \in \R^{n \times n}$ is an M-matrix if it is a Z-matrix and can be expressed in the form $A = 
		sI - B$, where $B \in \Rnn_+$ and 
		$s \ge \rho(B)$. 
		%
		We shall make use of the following results on the 
		properties of a nonsingular M-matrix  \cite{plemmons1977m}. 
		
		\begin{lem}\label{lem_M_Matrix}
			Let $A\in\mathbb{R}^{n\times n}$ be a Z-matrix. Then, 
			$A$ is a nonsingular M-matrix 
			if and only if one of the following conditions holds: 
			
			\begin{itemize}
				
				
				\item[(a)] $A$ is inverse-positive, i.e., $\exists A^{-1} 
				\in \R^{n \times n}_+$. 
				
				
				
				
				
				
				\item[(b)] $\exists \xx>\0$ with $A\xx > \0$. 
				
			\end{itemize}
		\end{lem}
		
		
				
				
				

\section{Model and Formulation}
	\label{sec:Formulation}
	
Suppose that the overall system consists of $N$ systems, and let $\cA = \{1, \ldots, N\}$ denote the set of systems. 		The set $\cA$ can be partitioned into $\{ \cA_p, \cA_n \}$.  Each system $i \in \cA_p \subsetneq \cA$ experiences external attacks from malicious actors in accordance with a Poisson process with rate $\lambda_i > 0$. On the other hand, the systems in $\cA_n$ do not experience external attacks and, hence, $\lambda_i =0$ for all $i \in \cA_n$. When a system experiences an attack, it suffers an infection and subsequent economic losses with some probability, called {\em breach probability}. 

In addition to external attacks from malicious actors, systems also experience {\em secondary} attacks from other infected systems. Thus, even the nodes in $\cA_n$ can experience secondary attacks. When a system suffers a secondary attack, it becomes infected with the same breach probability mentioned above. In other words, the breach probability is the same whether the attack is external or secondary. 

The breach probability of system $i$ depends on the security investment on the system: let $s_i \in \R_+$ be the security investment on system $i$ (e.g., investments in monitoring and diagnostic tools). The breach probability of system $i$ is determined by some function $q_i: \R_+ \to (0, 1]$. In other words, when the operator invests $s_i$ on system $i$, its breach probability is equal to $q_i(s_i)$. We assume that $q_i$ is decreasing, strictly convex and continuously differentiable for all $i \in {\cal A}$. It was shown~\cite{Bary-WEIS12} that, under some conditions, the breach probability is decreasing and log-convex (hence, strictly convex).

This can model the spread of virus/malware or failures in complex systems. The rate at which the infection of system $i$ causes that of another system $j$ is denoted by $\beta_{i,j} \in \R_+$. When $\beta_{i, j} > 0$, we say that system $i$ supports system $j$ or system $j$ depends on system $i$. Let $B = [b_{i,j} : i, j \in \mathcal{A}]$ be an $N \times N$ matrix that describes the infection rates among systems, where the element $b_{i,j}$ is equal to $\beta_{j,i}$. We adopt the convention $\beta_{i,i} = 0$ for all $i \in {\cal A}$. 

When system $i$ falls victim to an attack and becomes infected, the operator incurs costs at a certain rate for recovery (e.g., inspection and repair of servers). Besides recovery costs, the infection of system $i$ may also cause economic losses if, for example, some servers in system $i$ have to be taken offline for inspection and repair and are inaccessible during the period to other systems that depend on the servers. To model the recovery costs and economic losses, we assume that the infection of system $i$ causes total losses of $c_i$ per unit time. Recovery times of system $i$ following its infections are modeled using independent and identically distributed exponential random variables with parameter $\delta_i > 0$. Furthermore, the recovery times of different systems are mutually independent.

\subsection{Model}  
    \label{subsec:Model}
    
Define a directed graph ${\cal G} = ({\cal A}, {\cal E})$, where a directed edge from system $i$ to system $j$, denoted by $(i, j)$, belongs to the edge set ${\cal E}$ if and only if $\beta_{i,j} > 0$. Unlike in our earlier studies \cite{mai_globecom20, mai_ToN21}, here we do not assume that matrix $B$ is irreducible; instead, we assume that the graph ${\cal G}$ is only weakly connected, but not strongly connected.\footnote{If the graph $\mathcal{G}$ is not weakly connected, we can consider each weakly connected component of $\mathcal{G}$ separately.} We say that node $i$ is {\em exposed} to attackers if there is a directed path from some node $j \in \cA_p$ to node $i$, and denote the set of exposed nodes by $\mathcal{A}_E$. We call the remaining nodes {\em accessible} to the attackers and denote the set of accessible nodes by $\mathcal{A}_A$.

Let us comment on the distinction between `exposed' versus `accessible' systems. In our model, exposed systems may represent systems known to malicious actors, which communicate or exchange information with each other frequently. Thus, the attackers can target them either directly via external attacks or indirectly through secondary attacks. On the other hand, communication or exchange of information between accessible systems and exposed systems is highly asymmetric; most of communication is from accessible systems to exposed systems. Due to very limited communication from exposed systems to accessible systems, their infection rates by exposed systems are difficult to estimate reliably. For this reason, the infection rates are set to zero in our model even though they can still be vulnerable to occasional infections. As we will show, such systems can still serve as reservoir for infection in that the stable equilibrium for some of them may be endemic and they will be at `infected' state with strictly positive probability at steady state and, thus, can infect other exposed systems. For this reason, it is important to model such accessible systems.  

A weakly connected network can be partitioned into a set of maximally strongly connected components (MSCCs) $\{C_1, C_2, \ldots, C_m \} =: \cV\sC$ \cite{Khanafer-CDC14}.\footnote{A subgaph of a directed graph is maximally strongly connected if (a) it is strongly connected and (b) adding another node leads to a subgraph that is no longer strongly connected.} The set $\mathcal{A}_E$ comprises a subset of MSCCs, and $\mathcal{A}_A$ includes the remaining MSCCs. Based on this observation, we construct another directed graph ${\cal G}\sC = ({\cal V}\sC, {\cal E}\sC)$, where a directed edge from $v$ to $v'$ in ${\cal E}\sC$ indicates $\beta_{i,j} > 0$ for some system $i$ in $v$ and another system $j$ in $v'$. Since the vertices in ${\cal V}\sC$ are MSCCs, there is no directed cycle in ${\cal G}\sC$, i.e., ${\cal G}\sC$ is a directed acyclic graph (DAG). Using this property, we can show that ${\cal V}\sC$ has leaf vertices with no incoming edges; if this were not true, every vertex in $\cV\sC$ has an incoming edge. Since $\cV\sC$ is a finite set, this implies that there is a directed cycle, which contradicts the assumption that the vertices in $\cV\sC$ are MSCCs. 
			
We partition the vertex set ${\cal V}\sC$ into $\{ {\cal V}_0, \ldots, {\cal V}_m \}$, where ${\cal V}_\ell$, $\ell = 0, 1, \ldots, m$, is the set of MSCCs whose {\em maximum} distance from leaf vertices is $\ell$.\footnote{Here, the maximum distance from leaf vertices is defined to be the maximum among the maximum distances from the leaf vertices, i.e., the length of the longest path from any leaf vertex.} Obviously, ${\cal V}_0$ is the set of leaf vertices. Note that there is no directed edge coming into any vertex in ${\cal V}_k$ from any other vertex in ${\cal V}_\ell$, $\ell \geq k$. Hence, the frequency of attacks experienced by a system $i$ that belongs to some MSCC in ${\cal V}_k$ depends only on the states of systems in $\cup_{\ell = 0}^{k-1} {\cal V}_\ell$ along with $\lambda_i$.

\subsection{Dynamics and Equilibria}
    \label{subsec:DynamicsEqmb}
    
We adopt a similar framework used in \cite{mai_globecom20, mai_ToN21} and use the SIS model to capture the evolution of the system state. Let $p_i(t)$ be the probability that system $i$ will be at the `infected' state $(I)$ at time $t \in \R_+$, and define $\bp(t) := (p_i(t) : i \in \cA)$. The dynamics of $\pb(t)$ are approximated by the following (Markov) differential equations for $t \in \R_+$:
\begin{align} 
& \dot{\bp}(t)
= (\1 - \bp(t)) \circ \qb(\bs) \circ \big(
\blambda + B \bp(t) \big) 
    - \bdelta \circ \bp(t) 
    \label{eq:pdot_whole}
\end{align} 
where $\bp(0)\in [0,1]^N$,  $\bs := (s_i : i \in \cA)$ is the security investment vector, and $\qb(\bs) = (q_i(s_i) : i \in \cA)$ is the corresponding breach probability vector.
		
Suppose that for each security investment vector $\bs$, $\pb(t)$ converges to a unique stable equilibrium $\bar{\bp}(\bs)$ (the existence and uniqueness of such an equilibrium will be addressed in the subsequent section). Since the unique stable equilibrium of the differential system in \eqref{eq:pdot_whole} specifies the probability that each system will be infected at steady state, the average cost of the system is given by 
\[
F(\bs) := w(\bs) + \cb\T \bar{\pb}(\bs),
\]
where $w(\bs)$ is the cost of investing $\bs$ in the security of the systems (e.g., $w(\bs) = \sum_{i \in \cA} s_i$), and $\cb := (c_i : i \in \cA)$. 
We are interested in solving the following  problem: 
\begin{align} \label{eq:constr1}
F^* := \min_{\bs \in \cS} \  F(\bs)
    = \min_{\bs \in \cS}  \ \big( w(\bs) 
	+ \cb\T \bar{\pb}(\bs) \big)
\end{align}
where $\cS \subset \R_+^N$ is the feasible set for $\bs$. Throughout the paper, we assume that the cost function $w$ is a convex function and the feasible set $\cS$ is a convex set. 

The main challenge in solving this problem is that $\bar{\pb}(\bs)$ does not have a closed-form expression and need not be convex. For this reason, we consider the following alternative formulation with a higher dimension.
From \eqref{eq:pdot_whole}, for fixed $\bs$, $\bar{\bp}(\bs) \in [0,1]^N$ is a solution to the equation:
\[
(\1 - \bp) \circ \qb(\bs) \circ \big(\blambda + B \bp \big) 
= \bdelta \circ \bp
\]
Since $\qb(\bs) > \0$, the above equation is equivalent to
\begin{equation}    \label{eq:eqbm_cond1}
(\1 - \bp) \circ \big(
\blambda +  B \bp \big) 
- \qb(\bs)^{-1} \circ\bdelta \circ \bp = \0 , 
\end{equation} 
where $\qb(\bs)^{-1} := (q_i(s_i)^{-1} : i \in \cA)$. 
Now, \eqref{eq:constr1} can be reformulated in a following more explicit form:
\begin{align*} 
\hspace{-0.45in}{\bf [P]} \hspace{0.5in} 
\min_{\bs \in \cS, \bp \in [0,1]^N} & \ 
f(\bs,\bp) := w(\bs) + \cb\T \pb \lb
\mbox{subject to} & \quad  \eqref{eq:eqbm_cond1}
\end{align*}
Unfortunately, depending on $\blambda$ and $B$, for fixed $\bs$, there may be more than one $\pb$ that satisfies \eqref{eq:eqbm_cond1}, rendering the problem nonconvex. Note that this problem does not arise when $B$ is irreducible (i.e., $\mathcal{G}$ is strongly connected) and $\blambda \gneq \0$ (as considered in \cite{mai_ToN21}) because the uniqueness of the solution is already ensured. However, this is not the case when $\mathcal{G}$ is only weakly connected.

Our main idea to tackling this issue is as follows. For our problem {\bf [P]}, we are interested in a solution of \eqref{eq:eqbm_cond1} which is a stable equilibrium of \eqref{eq:pdot_whole}. 
%
Even when the stable equilibrium is unique, explicitly computing the unique stable equilibrium to carry out the optimization in \eqref{eq:constr1} does not lead to a computationally efficient approach because the optimization problem is not convex. 
In order to skirt this issue, in the following section, we propose a more practical approach based on perturbed problems, which then allows us to leverage the efficient algorithms proposed in \cite{mai_ToN21}.

\section{Perturbed Problem and Main Results}
    \label{sec:Perturbed}

In this section, we describe how we can find a good solution to the problem in \eqref{eq:constr1} in a computationally efficient manner. To this end, we first construct a new approximated problem by perturbing the attack arrival rate $\blambda$ by adding a nonnegative vector $\bvarepsilon \gneq \0$. Although we can work with any nonnegative vector $\bvarepsilon$ such that a unique stationary vector $\bar{\pb}$ is strictly positive, in order to simplify our exposition, we assume that the perturbation vector $\bvarepsilon$ takes the form $\bvarepsilon = \epsilon \1$ with $\epsilon > 0$.\footnote{As long as we perturb the attack arrival rate of at least one system in each accessible MSCC (by $\epsilon$), our results continue to hold.} In other words, we perturb the external attack rate of every system by $\epsilon$. For fixed $\epsilon \geq 0$, we define $\blambda^{\epsilon} := \blambda + \epsilon \1$. 

Suppose that the security investment $\bs$ is fixed. Then, 
for all $\epsilon > 0$, there is unique $\bar{\bp}^{\epsilon}(\bs) > \0$ that satisfies the following~\cite{mai_ToN21}:
\begin{equation}    \label{eq:eqbm_cond2}
(\1 - \bp) \circ \big(
\blambda^{\epsilon} +  B \bp \big) 
- \qb(\bs)^{-1} \circ \bdelta \circ \bp = \0
\end{equation} 
Define 
\begin{align} 
F_{\epsilon}^* := \min_{\bs \in \cS}\quad  
\big( w(\bs) + \cb\T \bar{\bp}^{\epsilon}(\bs) \big) .
\label{eq:PertProb}
\end{align}
Our approach is as follows: first, we know that, for fixed $\epsilon > 0$, $\bar{\bp}^{\epsilon}(\bs)$ is continuous in $\bs$~\cite{mai_ToN21}. We will prove that, for fixed $\bs$, $\bar{\bp}^{\epsilon}(\bs)$ is continuous in $\epsilon \geq 0$. From the continuity of the objective function, this implies that $F^\star_{\epsilon}$ is continuous in $\epsilon \geq 0$. Finally, we can solve the perturbed problem for some small $\epsilon$, and use the solution to the perturbed problem as an approximated solution to the original problem with $\epsilon = 0$. 


We note that when $\epsilon = 0$, there could be multiple solutions to \eqref{eq:eqbm_cond1} (or, equivalently,  \eqref{eq:eqbm_cond2} with $\epsilon = 0$). Therefore, in order to make use of this observation, we need to establish that $\lim_{\epsilon \to 0} \bar{\pb}^\epsilon(\bs)$ exists and coincides with the {\em unique} stable equilibrium of \eqref{eq:eqbm_cond1}. To this end, the following lemma establishes the existence of a unique stable equilibrium of \eqref{eq:eqbm_cond1}. 


\begin{prop}   \label{prop:unique_eqbm}
For fixed $\bs \in \cS$, there is a unique stable equilibrium $\bar{\pb}(\bs)$ of \eqref{eq:pdot_whole}, which satisfies \eqref{eq:eqbm_cond1}.
\end{prop}
\begin{proof}
For each MSCC $v$ in $\cV\sC$, we denote by $B_{v}$ the corresponding submatrix of $B$ (only with columns and rows of $B$ associated with the nodes in $v$). Similarly, given a vector (e.g., $\qb(\bs), \bdelta, \blambda$), we use the subscript $v$ to denote the subvector associated with nodes in $v$ (e.g., $\qb_v(\bs), \bdelta_v, \blambda_v$).

We prove the proposition by induction: 
\myskip

\noindent{\bf Step 1:} 
Consider an MSCC $v$ in $\cV_0$. Then, since $v$ is strongly connected and there is no incoming edge to $v$, there is a unique stable equilibrium to which $\pb_v(t)$ converges, starting with any $\pb_v(0) \neq \0$~\cite{mai_ToN21}: (i) If $v \subset \mathcal{A}_E$, then there is a strictly positive unique stable equilibrium $\bar{\pb}_v(\bs)$. (ii) For $v \subset \mathcal{A}_A$, there are two possibilities -- (a) if $\rho({D_{v}}^{-1} B_{v}) > 1$, where $D_{v} = \texttt{diag}(\qb_{v}^{-1}(\bs) \circ \bdelta_{v})$, there are two equilibria -- one strictly positive stable equilibrium $\bar{\pb}_{v}(\bs)$ and an unstable equilibrium $\0$; and (b) if $\rho({D_{v}}^{-1} B_{v}) \leq 1$, $\bar{\pb}_{v}(\bs) = \0$ is the unique stable equilibrium. 
\myskip

\noindent {\bf Step 2:} 
Assume that there is a unique stable equilibrium for all MSCCs in $\cup_{\ell=0}^k \cV_\ell =: \tilde{\cV}_k$, $k \in  \{0, \ldots, m-1\}$. Then, $\pb_{v'}(t)$ converges to $\bar{\pb}_{v'}(\bs)$ for all $v' \in \tilde{\cV}_k$. This implies that, for every system $i$ in some MSCC $v \in \cV_{k+1}$, its total attack rate, including both external attacks from malicious actors and secondary attacks from other nodes not in MSCC $v$, converges to $\tilde{\lambda}_i(\bs) \geq \lambda_i$ as $t \to \infty$. Hence, we can assume that the new attack rates of the systems in $v \in \tilde{\cV}_{k+1}$ are given by $(\tilde{\lambda}_i(\bs) : i \in v)$ and follow the same argument described in Step 1 above to find the unique stable equilibrium for each MSCC $v \in \cV_{k+1}$. Since this is true for all $k = 0, \ldots, m-1$, this tells us that there is a unique stable equilibrium $\bar{\pb}(\bs)$ for the overall system. 
\end{proof}

\subsection{Continuity of Stable Equilibrium $\bar{\pb}$}

For $\bs \in \cS$, define $\gb^\bs:\mathbb{R}^{n+1} \to \mathbb{R}^n$, where
\begin{equation}    \label{eq:g}
\gb^\bs(\epsilon, \pb) = (\1 - \bp) \circ \big(
	\blambda^{\epsilon} + B \bp \big) 
		- \qb(\bs)^{-1} \circ \bdelta \circ \bp  . 
\end{equation}
Clearly, $\gb^\bs$ is a continuously differentiable function and its partial derivative w.r.t. $\pb$ is given by
$$
\partial_{\pb} \gb^\bs(\epsilon, \pb) 
\! = \! \diag(\1 - \pb) B - \diag(\qb(\bs)^{-1} \circ\bdelta 
	\! + \! \blambda^{\epsilon} \! + \! B \pb) .
$$

\begin{prop}    \label{prop:Pess_continuity}
For each  $\bs \in \cS$, $\bar{\bp}^{\epsilon}(\bs)$ is continuous in $\epsilon > 0$.
\end{prop}

\begin{proof}		
In order to prove the proposition, we make use of the following lemma. 
\begin{lem}   \label{lemma_M_M_matrix}
The matrix $-\partial_{\pb} \gb^\bs(\epsilon, \bar{\bp}^{\epsilon}(\bs))$ is a nonsingular M-matrix for any $\epsilon > 0$.
\end{lem}
\begin{proof}
We denote $\bar{\bp}^{\epsilon}(\bs)$ by $\bar{\bp}^{\epsilon}$ for notational simplicity. Let $M = -\partial_{\pb} \gb^\bs(\epsilon, \bar{\bp}^{\epsilon})$. Clearly, $M$ is a Z-matrix. Moreover, 
\beqan
M \bar{\pb}^\epsilon 
\myeq - (\1 - \bar{\pb}^\epsilon) \circ B \bar{\pb}^\epsilon
    + \qb(\bs)^{-1} \circ \bdelta \circ \bar{\pb}^\epsilon
    + \blambda^\epsilon \circ \bar{\pb}^\epsilon \\
&& + B \bar{\pb}^\epsilon \circ \bar{\pb}^\epsilon \\
\myeq - (\1 - \bar{\pb}^\epsilon) \circ (\blambda^\epsilon + B \bar{\pb}^\epsilon)
    + \qb(\bs)^{-1} \circ \bdelta \circ \bar{\pb}^\epsilon \\ 
&& + B \bar{\pb}^\epsilon \circ \bar{\pb}^\epsilon 
    + \blambda^\epsilon  \\
\myeq B \bar{\pb}^\epsilon \circ \bar{\pb}^\epsilon 
    + \blambda^\epsilon, 
\eeqan
where the last equality follows from $\gb^\bs(\epsilon, \bar{\pb}^\epsilon) = \0$. Since $\blambda^{\epsilon} > \0$ and $\bar{\bp}^{\epsilon} > \0$ for $\epsilon > 0$, we have $\blambda^{\epsilon} + \bar{\bp}^{\epsilon} \circ (B \bar{\bp}^{\epsilon} ) > \0$. Thus, Lemma~\ref{lem_M_Matrix}-(b) implies that $M$ is a nonsingular M-matrix. 
\end{proof}

Since $\partial_{\pb} \gb^\bs(\epsilon, \bar{\bp}^{\epsilon}(\bs))$ is invertible by Lemma~\ref{lemma_M_M_matrix}, the implicit function theorem tells us that $\bar{\bp}^{\epsilon}(\bs)$ is continuous in $\epsilon > 0$. This proves the proposition. 
\end{proof}

For $\bs \in \cS$, define a mapping $\bar{\blambda}^\bs: \R_+ \to \R_+^N$, where, for each $v \in \cV\sC$, 
\begin{equation*}
\bar{\blambda}^\bs_v(\epsilon) 
= \begin{cases}
    \blambda^{\epsilon}_v + \big( B_{-v} \bar{\pb}^{\epsilon}_{-v}(\bs) \big)_v & \mbox{if } \epsilon > 0, \\
    \blambda_v + \big( B_{-v} \bar{\pb}_{-v}(\bs) \big)_v & \mbox{if } \epsilon = 0,  
\end{cases}
\end{equation*}
where $B_{-v}$ is a submatrix of $B$ without the columns corresponding to the systems that belong to the MSCC $v$, and $\bar{\pb}^{\epsilon}_{-v}(\bs)$ and $\bar{\pb}_{-v}(\bs)$ are the subvectors of $\bar{\pb}^{\epsilon}(\bs)$ and $\bar{\pb}(\bs)$, respectively, obtained after removing the elements for the systems in $v$. Obviously, $\bar{\blambda}^\bs(\epsilon)$ tells us the total attack rates, including both external attacks from malicious actors and secondary attacks coming from other systems that do not belong to the same MSCC, at the unique stable equilibrium as a function of the security investments $\bs$ and perturbation $\epsilon$. 

We find it convenient to define the following maps: for each $v \in \cV\sC$, let $n_v := | v |$. For each $\bs \in \cS$, define $\hat{\gb}^\bs_v : \R^{n_v+1} \to \R^{n_v}$, where
\begin{eqnarray}
\ \hat{\gb}^\bs_v(\epsilon, \pb_{v}) 
\! = \! (\1 \! - \! \pb_{v}) \! \circ \! \big( \bar{\blambda}^\bs_v(\epsilon) \! + \! B_v \pb_v \big) 
    \! - \! \qb_v(\bs)^{-1} \! \circ \! \bdelta_v \! \circ \! \pb_v, 
    \label{eq:ghat}
\end{eqnarray}
where $B_v$ is the $n_v \times n_v$ submatrix of $B$ with columns and rows corresponding to the systems in $v$. Clearly, $\hat{\gb}^\bs_v$ is the mapping $\gb$ defined in \eqref{eq:g} restricted to the MSCC $v$ with fixed attack rates $\bar{\blambda}_v^\bs(\epsilon)$ coming from outside the MSCC $v$. We point out that, from the viewpoint of a system, there is no distinction between an external attack from a malicious actor or a secondary attack coming from another system in a different MSCC. Clearly, for fixed $\epsilon > 0$, $\bar{\pb}^{\epsilon}(\bs)$ is the unique solution that satisfies
\beqa
\hat{\gb}^\bs_v \big( \epsilon, \bar{\bp}^{\epsilon}_v(\bs) \big) = \0 
\quad \mbox{for all } v \in \cV\sC. 
    \label{eq:cond1}
\eeqa
Similarly, $\bar{\pb}(\bs)$ satisfies
\beqa
\hat{\gb}^\bs_v\big( 0, \bar{\bp}_v(\bs) \big) = \0 \quad \mbox{for all } v \in \cV\sC. 
    \label{eq:cond2}
\eeqa

\begin{prop}    \label{prop:continuity2}
Fix $\bs \in S$. Then, for any decreasing positive sequence $\{\epsilon_l : l \in \N\}$ with $\lim_{l \to \infty} \epsilon_l = 0$, we have $\lim_{l \to \infty} \bar{\pb}^{\epsilon_l}(\bs) = \bar{\pb}(\bs)$.
\end{prop}
\begin{proof} 
Denote $\bar{\pb}_v(\bs)$ and $\bar{\pb}_v^\epsilon(\bs)$ by $\bar{\pb}_v$ and $\bar{\pb}_v^\epsilon$, respectively. Following a similar argument and the notation used in the proof of Proposition~\ref{prop:unique_eqbm}, we will prove the proposition by induction, starting with MSCCs in $\cV_0$:
\myskip

\noindent
\underline{\bf Step 1:} Consider an MSCC $v$ in $\cV_0$. Note that,  since there is no incoming edge to $v$ from any other node in $\cG\sC$, for any $\pb \in [0, 1]^N$, we have $(B \pb)_v = B_v \pb_v$ (or, equivalently, $(B_{-v} \pb_{-v})_v = \0$) and, thus, $\bar{\blambda}^\bs_v(\epsilon) = \blambda_v^\epsilon$ for all $\epsilon \geq 0$. We examine the following two possible cases separately.

{\bf Case 1a. $\bar{\pb}_{v} > \0$: } In this case, the continuity of $\bar{\pb}_v^\epsilon(\bs)$ at $\epsilon = 0$ follows from the implicit function theorem because the Jacobian matrix $\partial_{\pb_{v}} \hat{\gb}_v(0, \bar{\pb}_v)$ is nonsingular as follows:
\beqan
&& \myhb M^\bs_v 
:= - \partial_{\pb_v} \hat{\gb}^\bs_v(0, \bar{\pb}_v) \lb
\myeq - \diag(\1 - \bar{\pb}_v) B_v 
    + \diag\big( \blambda_v
    \! + \! B_v \bar{\pb}_v \! + \! \qb_v(\bs)^{-1} \! \circ \! \bdelta_v \big)
\eeqan
is a Z-matrix. Furthermore, 
\beqan
M^\bs_v \bar{\pb}_v 
\! \myeq \! - (\1 - \bar{\pb}_v) \circ B_v \bar{\pb}_v
    \! + \! \qb_v(\bs)^{-1} \circ \bdelta_v \circ \bar{\pb}_v
    \! +  \blambda_v \circ \bar{\pb}_v \\
&& \! +  B_v \bar{\pb}_v \circ \bar{\pb}_v \\
\myeq \underbrace{- (\1 - \bar{\pb}_v) \circ (\blambda_v + B_v \bar{\pb}_v)
    + \qb_v(\bs)^{-1} \circ \bdelta_v \circ \bar{\pb}_v}_{= - \hat{\gb}^\bs_v(0, \bar{\pb}_v)} \\
&& \! +  B_v \bar{\pb}_v \circ \bar{\pb}_v 
    \! + \!\blambda_v \\
\myeq B_v \bar{\pb}_v \circ \bar{\pb}_v 
    + \blambda_v, 
\eeqan
where the last equality follows from $\hat{\gb}^\bs_v(0, \bar{\pb}_v) = \0$. Since $\bar{\bp}_v > \0$ and $B_v$ is irreducible (as $v$ is strongly connected), we have $B_v \bar{\pb}_v \circ \bar{\pb}_v > \0$. Thus, Lemma~\ref{lem_M_Matrix}-(b) tells us that $M^\bs_v$ is a nonsingular M-matrix. 

{\bf Case 1b. $\bar{\pb}_{v} = \0$: } In this case, $\bar{\pb}_v$ is the unique solution to \eqref{eq:cond2}
(and $\bar{\pb}^\epsilon_v$ is the unique solution to \eqref{eq:cond1} for all $\epsilon > 0$). Because $\hat{\gb}^\bs_v$ defined in \eqref{eq:ghat} is continuous and $\bar{\blambda}^\bs_v(\epsilon_l) = \blambda^{\epsilon_l}_v \to \blambda_v$ as $l \to \infty$, it follows that $\lim_{l \to \infty} \bar{\pb}^{\epsilon}_v = \bar{\pb}_v$. 
\myskip


\noindent
\underline{\bf Step 2:} Assume $\lim_{l \to \infty} \bar{\pb}_{v'}^{\epsilon_l}(\bs) = \bar{\pb}_{v'}(\bs)$ for every MSCC $v' \in \cup_{\ell=0}^k \cV_\ell =: \tilde{\cV}_k$, $k \in  \{0, \ldots, m-1\}$. This means that, for all $v' \in \tilde{\cV}_\ell$, $\ell = 0, \ldots, k$, we have (a) $\blambda^{\epsilon_l}_{v'} \to \blambda_{v'}$ and (b) $\bar{\pb}^{\epsilon_l}_{v'}(\bs) \to \bar{\pb}_{v'}(\bs)$ as $l \to \infty$. As a result, for every $v \in \cV_{k+1}$, $\bar{\blambda}^\bs_v(\epsilon) \to \bar{\blambda}^\bs_v(0)$ as $l \to \infty$. We can show that, for all $v \in \cV_{k+1}$, $\lim_{l \to \infty} \bar{\pb}^{\epsilon_l}_v(\bs) = \bar{\pb}_v(\bs)$ by following an analogous argument used in Step 1. We consider the following two cases:

{\bf Case 2a. (i) $\bar{\blambda}_v^\bs(0) \neq \0$ or (ii) $\bar{\blambda}_v^\bs(0) = \0$ and $\rho(D_v^{-1} B_{v}) > 1$: } In this case, $\bar{\pb}_v(\bs)$ is the unique positive solution to \eqref{eq:cond2}. Thus, the claim follows from the implicit function theorem; the Jacobian matrix $-\partial_{\pb_v} \hat{\gb}^\bs_v(0, \bar{\pb}_v)$ can be shown to be nonsingular using the same argument employed in Case 1a of Step 1 above. 

{\bf Case 2b. $\bar{\blambda}_v^\bs(0) = \0$ and $\rho(D_v^{-1} B_{v}) \leq 1$: } We know $\bar{\pb}_v = \0$ is the unique solution to \eqref{eq:cond2}. Because $\hat{\gb}^\bs_v$ in \eqref{eq:ghat} is continuous and $\hat{\gb}^\bs_v(\epsilon, \bar{\pb}^\epsilon) = \0$, it follows that $\lim_{l \to \infty} \bar{\pb}^{\epsilon_l}_v = \bar{\pb}_v$. 
\end{proof}

One would expect that the minimum cost we can achieve by solving \eqref{eq:PertProb} would not decrease with $\epsilon$. We will show that this is indeed the case. 

\begin{prop}\label{prop_p_increasing_in_eps}
For each  $\bs \in \cS$, $\bar{\bp}^{\epsilon}(\bs)$ is increasing in $\epsilon > 0$.
\end{prop}
\begin{proof}
The implicit function theorem used in the proof of  Proposition~\ref{prop:Pess_continuity} also tells us that, for all $\epsilon > 0$, 
\beqan
\partial_{\epsilon} \bar{\pb}^\epsilon(\bs)
\myeq \left[ - \partial_{\bar{\pb}} \gb^{\bs}\big( \epsilon, 
    \bar{\pb}^\epsilon(\bs) \big) \right]^{-1} 
    \partial_\epsilon \gb^{\bs}\big( \epsilon,                  \bar{\pb}^\epsilon(\bs) \big) \lb
\myeq \left[ - \partial_{\bar{\pb}} \gb^{\bs}\big( \epsilon, 
    \bar{\pb}^\epsilon(\bs) \big) \right]^{-1} 
    \big( \1 - \bar{\pb}^\epsilon(\bs) \big) . 
\eeqan
As $- \partial_{\bar{\pb}} \gb^{\bs}\big( \epsilon, \bar{\pb}^\epsilon(\bs) \big)$ is a nonsingular M-matrix (Lemma~\ref{lemma_M_M_matrix}), its inverse-positivity (Lemma~\ref{lem_M_Matrix}-(a)) tells us $\partial_{\epsilon} \bar{\pb}^\epsilon(\bs) > \0$.
\end{proof}

\begin{cor}    \label{coro:1}
The optimal value $F^*_\epsilon$ is increasing in $\epsilon \geq 0$.
\end{cor}

Corollary~\ref{coro:1}, together with the continuity of the objective functions in \eqref{eq:constr1} and \eqref{eq:PertProb} and  Propositions~\ref{prop:Pess_continuity} and \ref{prop:continuity2}, tells us that as we reduce $\epsilon$, the optimal value $F^*_\epsilon$ of the perturbed problem will decrease to the optimal value $F^*$ of the original problem. Therefore, this suggests that if we solve the perturbed problem \eqref{eq:PertProb} with sufficiently small $\epsilon$, the optimal point we obtain will likely be a good solution for the original problem, and the optimal value of \eqref{eq:PertProb} will serve as an upper bound to $F^*$.

\section{Solving the Perturbed Problem}
    \label{sec:Proposed}

In this section, we discuss how we can solve the perturbed problem in \eqref{eq:PertProb} by extending the formulation and adapting the algorithms proposed in \cite{mai_ToN21}. First, we rewrite the perturbed problem as follows:
\begin{align}
\hspace{-0.06in}{\bf [PP]} \hspace{0.25in}
\min_{\bs, \bp} &\quad f(\bs, \pb) = w(\bs) + \cb\T\pb \lb
\mathsf{s.t.} &\quad (\bp^{-1} - \1) \circ \big(
\blambda^{\epsilon} +  B \bp \big) 
= \qb(\bs)^{-1}\circ\bdelta \label{eq:eqbm_cond_pinv}\\
&\quad \bs \in \mathcal{S}, \quad \bp > 0 \nnb
\end{align}
Note that constraint \eqref{eq:eqbm_cond_pinv} is the same as \eqref{eq:eqbm_cond2} when $\bp >0$. Next, we introduce the following change of variables:
$$
    d_i = (q_i(s_i))^{-1} \ \text{ or, equivalently, } \
s_i = q_i^{-1}(d_i^{-1}), 
$$
where $q_i^{-1}$ is the inverse map of $q_i$. Let $\mathcal{D} = \{ \bq(\bs)^{-1} \ | \ \bs \in \mathcal{S} \}$.  
As a result, we obtain the following equivalent problem:
\begin{align}
\hspace{-0.2in}{\bf [EP]} \hspace{0.4in} \min_{\bd, \bp} &\quad \tilde{f}(\bd, \pb) 
    = \tilde{w}(\bd) + \cb\T\pb \lb
\mathsf{s.t.} &\quad (\bp^{-1} - \1) \circ \big(
\blambda^{\epsilon} +  B \bp \big) 
= \bd\circ\bdelta \label{eq:eqbm_cond_pinv2}\\
&\quad \bd \in \mathcal{D}, \quad \bp > 0 \nnb
\end{align}
where $\tilde{w}(\bd):= w(\bq^{-1}(\bd^{-1}))$. 
This problem is nonconvex in general because of the cost function, the equality constraint in \eqref{eq:eqbm_cond_pinv2}, and possibly the constraint set $\mathcal{D}$. However, as shown in \cite{mai_ToN21}, a local minimizer can be found efficiently using a reduced gradient method (RGM). This also provides an upper bound on $F^*_\epsilon$. 


Suppose that $\mathcal{D}$ is convex and $\tilde{w}(\bd)$ is convex in $\bd \in \mathcal{D}$, which  implies the convexity of $\tilde{f}(\bd, \pb)$. It can be shown that, provided that $w$ is convex and increasing, e.g., $w(\bs) = \1^T \bs$, the second assumption holds for a family of breach probability functions $q_i(s_i) = (1 + \kappa_i s_i)^{-\beta_i}$ for some $\kappa_i>0$ and $\beta_i\in (0,1]$, in which case $s_i = (d_i^{1/\beta_i} - 1)\kappa_i^{-1}$ is convex in $d_i$. When such assumption does not hold, one might consider a suitable convex lower bound of $\tilde{w}(\bd)$ instead. In addition, when $\mathcal{S} = \{\bs \in \R_+^N \ | \ \1\T \bs \leq s_{{\rm budget}}\}$, where $s_{{\rm budget}}$ is the total budget, the constraint set $\mathcal{D}$ would be convex for the aforementioned family of breach probability functions. 

Following an approach analogous to \cite{mai_ToN21}, we can obtain a convex relaxation of the perturbed problem to deal with the nonconvex constraint in \eqref{eq:eqbm_cond_pinv}. Define the following variables: 
\begin{align} 
\hspace{-0.1in} \pb := e^{-\yb}, \ \tb := \blambda^{\epsilon}\circ e^{\yb}, \ 
U := \diag(e^{\yb})B\diag(e^{-\yb}).  \label{eqEXP_var}
\end{align}
Using these new variables, \eqref{eq:eqbm_cond_pinv2} can be rewritten as follows.
\begin{align}
\tb + U\1 
= \boldsymbol{\lambda}^{\epsilon} + B\pb + \bd \circ\boldsymbol{\delta}
\label{eqSteadyState6}
\end{align} 
Finally, we relax the equality constraints in \eqref{eqEXP_var} to the following convex inequality constraints.
\begin{align} 
\hspace{-0.1in}
\1 \ge \pb \ge e^{-\yb}, \  \tb \ge  \blambda^{\epsilon}\!\circ\! e^{\yb}, 
\ U \ge \diag(e^{\yb})B\diag(e^{-\yb}).
    \label{eqEXP_cone}
\end{align}
These yield the following convex relaxation of {\bf [EP]}:
\begin{align}
\hspace{-0.25in}{\bf [CR]} \hspace{0.4in} 
\min_{\bd, \bp, \by, \tb, U} &\quad \tilde{f}(\bd, \pb) = \tilde{w}(\bd) + \cb\T\pb \label{eq:CvxRelax} \\
\mathsf{s.t.} &\quad \eqref{eqSteadyState6}, \eqref{eqEXP_cone}, \ \bd \in \mathcal{D}, \ \by > \0 \nnb
\end{align}

\begin{thm} \label{thm:1}
Suppose $\bx^+_{\rm R} := (\bd^+, \pb^+, \yb^+, 
\tb^+, U^+)$ is an optimal point of {\bf [CR]}.
Then, we have
\beqa
\tilde{f}(\bd^+, \pb^+) 
\le F_{\epsilon}^* \le \tilde{f}(\bd', \pb'), 
    \label{eq:thm3}
\eeqa
where the pair $(\bd', \pb')$ is a feasible point of {\bf [EP]} given by
\[
\pb' = e^{-\yb^+} \ \mbox{ and } \
\bd' = \bd^+ + \diag(\boldsymbol{\bdelta}^{-1})
        B(\pb^+ - \pb').
\]  
In addition, {\bf [CR]} is exact, i.e., $(\bd', \pb')$ solves {\bf [EP]}, if 
\begin{align} 
B\T \diag(\boldsymbol{\delta}^{-1}) \nabla \tilde{w}(\bd) 
    \le  \cb  \ \mbox{ for all } \bd \in \mathcal{D}. 
    \label{eq:thm1-cond}
\end{align}
\end{thm}
\begin{proof}
The first inequality in \eqref{eq:thm3} is obvious because {\bf [CV]} is a relaxation of {\bf [EP]}. The second inequality follows from the feasibility of  $(\bd', \pb')$ as shown below. 

$\bullet$ {\bf Proof of feasibility of $(\bd', \pb')$ for {\bf [EP]}: } 
Since $\bd^+ \ge \qb(\0)$ and $\pb^+ \ge e^{-\yb^+} = 
\pb' > \0$, it follows that $\bd' \ge \qb(\0)$. It remains 
to show that $(\bd', \pb')$ satisfies \eqref{eq:eqbm_cond_pinv2}:
\begin{equation}
(\pb')^{-1}\circ\boldsymbol{\lambda}^\epsilon 
    + (\pb')^{-1}\circ B\pb' 
= \boldsymbol{\lambda}^\epsilon + B\pb' 
    + \bd' \circ \bdelta.
    \label{eq:thm1-1}
\end{equation}
From the definition of $(\pb',\bd')$, 
\eqref{eq:thm1-1} is equivalent to 
\begin{align}
e^{\yb^+} \circ \boldsymbol{\lambda}^\epsilon 
    + e^{\yb^+}\circ Be^{-\yb^+} 
= \boldsymbol{\lambda}^\epsilon 
    + B\pb^+ + \bd^+ \circ \bdelta, 
    \label{eq:thm1-2}
\end{align}
where the right-hand side equals $\tb^+ + U^+\1$
from \eqref{eqSteadyState6}. 

To show that the equality in \eqref{eq:thm1-2} holds, we will prove that
\begin{equation} 
\tb^+ = \blambda^\epsilon \circ e^{\yb^+} \ \mbox{ and } \ 
U^+ = \diag(e^{\yb^+})B\diag(e^{-\yb^+}) 
    \label{eq:thm1-3}
\end{equation}
i.e., the inequality constraints of $\tb$ and 
$U$ are active at the solution $\bx^+_{\rm R}$ of 
{\bf [CR]}. 
Let $\bmu_p, \overline{\bmu}_p$, $\bmu_t$, 
$\bmu_y \in \RN_+$ and $\Phi \in \mathbb{R}^{N\times N}_+$ 
be the Lagrangian multipliers associated with inequality 
constraints, and $\boldsymbol{\sigma} \in \RN$ with 
the equality constraint in \eqref{eqSteadyState6} of 
{\bf [CR]}. Then, the Lagrangian function is 
given by 
\begin{align}
L = & \ \tilde{w}(\bd) + I_{\mathcal{D}}(\bd) + \langle\cb,\pb\rangle 
+ \langle \bmu_p, e^{-\yb}-\pb \rangle  \nnb\\
&+ \langle \bmu_t,\blambda^\epsilon \circ e^{\yb} 
    - \tb \rangle + \langle \Phi, \diag(e^{\yb})B\diag(e^{-\yb}) 
    - U\rangle \nnb\\
&+ \langle \overline{\bmu}_p, \pb-\1 \rangle 
    - \langle \bmu_y,\yb \rangle\nnb\\
&+ \langle \bsigma, \tb + U\1 
    - \boldsymbol{\lambda}^\epsilon - B\pb 
    - \bd \circ \bdelta \rangle, \nnb
\end{align}
where $I_{\mathcal{D}}(\cdot)$ is the set indicator function of the set $\mathcal{D}$, i.e., $I_{\mathcal{D}}(\bd) = 0$ if $\bd \in \mathcal{D}$ and $\infty$ otherwise. Note that we use the set indicator function for convenience; as will be clear below, we do not need to consider the optimal variables $\bd$.\footnote{If $\mathcal{D} = \{\bd~|~ \bh(\bd) \le \0\}$ for some convex function $h$, then we can also introduce an additional term $\langle\bmu_h, \bh(\bd)\rangle$ in the Lagrangian and the proof does not change.}
Since the problem is convex, the  
Karush-Kuhn-Tucker (KKT) conditions are 
necessary and sufficient for optimality. Thus, we have
\begin{subequations}
\begin{eqnarray}
\0 \myeq \nabla_p L 
    = \cb - \bmu_p + \overline{\bmu}_p - B\T\bsigma \label{L_p}\\
\0 \myeq \nabla_t L 
    = \bsigma - \bmu_t \label{L_t}\\
0 \myeq \partial_{u_{ij}} L 
    = \sigma_i - \phi_{ij} \ \mbox{ for all }
        i,j \in \mathcal{A} 
    \label{L_u}\\
0 \myeq \partial_{y_i} L 
    = \mu_{t_i} \lambda_i^\epsilon e^{y_i^+} 
    - \mu_{p_i}e^{-y_i^+} 
    + \textstyle\sum_{j \in \mathcal{A}}\phi_{ij}b_{ij} 
        e^{y_i^+ \!-y_j^+} \nnb\\
&& \hspace{-0.0in}
    - \textstyle\sum_{j \in \mathcal{A}} \phi_{ji}b_{ji} 
        e^{y^+_j\!-y^+_i} -\mu_{y_i} \ \mbox{ for all } 
        i \in \mathcal{A} 
        \label{L_y}\\
\0 \myeq \bmu_t \circ (\blambda^\epsilon \circ e^{\yb^+}-\tb^+) \label{slk_t}\\
0 \myeq \phi_{ij}(u^+_{ij} - b_{ij}e^{y^+_i-y^+_j})
    \ \mbox{ for all } i,j \in \mathcal{A}. 
    \label{slk_u}
\end{eqnarray}
\end{subequations}
From \eqref{L_t} and \eqref{L_u}, we get
\begin{equation}
\sigma_i = \mu_{t_i} = \phi_{ij} \ge 0 
    \ \mbox{ for all } i,j \in \mathcal{A}.
    \label{L_tu}
\end{equation}
Combining this and \eqref{L_y} yields 
\begin{eqnarray*}
&& \hspace{-0.2in}
\sigma_i (\lambda_i^\epsilon e^{y_i^+}          
    + \textstyle\sum_{j\in \mathcal{A}}b_{ij}e^{y_i^+-y_j^+})  \\
\myeq \mu_{p_i}e^{-y_i^+} \!+\! \mu_{y_i} 
    + \textstyle\sum_{j\in \mathcal{A}}\sigma_j b_{ji}e^{y_j^+ - y_i^+} \\
\mygeq (\mu_{p_i} + \textstyle\sum_{j\in \mathcal{A}}\sigma_j b_{ji}) 
    e^{-y_i^+}
\geq (c_i+\overline{\mu}_{p_i})e^{-y_i^+} 
> 0,
\end{eqnarray*}
where the first inequality follows from $\bmu_y\ge \0$ and $\yb^+\ge \0$, and the second from~\eqref{L_p}. Thus,  $\bsigma > \0$, which together with \eqref{L_tu} and the slackness conditions \eqref{slk_t} and \eqref{slk_u} implies that \eqref{eq:thm1-3} must hold. 

$\bullet$ {\bf Exactness of convex relaxation [CR]: } Since $\tilde{w}$ is assumed convex, we have $\tilde{w}({\bd}') - \tilde{w}(\bd^+) \le \nabla \tilde{w}({\bd}')\T ({\bd}' - \bd^+) = \nabla w({\bd}')\T \diag(\bdelta^{-1})B(\pb^+ - {\pb}')$, where the second equality follows from the definition of $\bd'$. From this inequality, the gap $\tilde{f}({\bd}', {\pb}') - \tilde{f}(\bd^+, \pb^+) = \tilde{w}(\bd') + \cb\T \pb' - (\tilde{w}(\bd^+) + \cb\T \pb^+) \leq (\nabla \tilde{w}({\bd}')\T \diag(\boldsymbol{\delta}^{-1})B - \cb\T)(\pb^+ - {\pb}')$. Under condition \eqref{eq:thm1-cond}, together with $\pb^+ \ge {\pb}'$, this gap is nonpositive. Thus, the inequality  in \eqref{eq:thm3} tells us that this gap must be zero, proving that {\bf [CR]} is exact.
\end{proof}

We end this section with the following remarks. First, note that condition~\eqref{eq:thm1-cond} can be checked prior to solving the relaxed problem [\textbf{CR}]. This can be done easily when $\nabla \tilde{w}$ (or an upper bound) is known and the constraint set $\mathcal{D}$ is simple enough. Second, our approach in this section is based on the reformulation \textbf{[PP]} of the original problem, where we use the condition that $\bp > \0$; in the relaxation [\textbf{CR}], this is ensured by imposing the condition $\bp \ge e^{\by}$. If $p_i = 0$ for some $i$ at an optimal point of the original problem (which can happen when $\mathcal{G}$ is only weakly connected and $\lambda_i = 0$), this condition would be violated and we cannot use \eqref{eq:eqbm_cond_pinv}; in the relaxation, this would correspond to having $y_i \to \infty$. As a result, our perturbation of the attack rates introduced in Section~IV proves to be meaningful in practice as it not only allows us to employ efficient methods to solve the problem approximately (as will be demonstrated numerically in the next section), but also  takes into account more stringent scenarios (with varying external attack rates).

\section{Numerical Examples}
    \label{sec:Numerical}
    
In this section, we provide some numerical results to evaluate the proposed method. Our numerical studies are carried out in MATLAB (version R2018b) on a laptop with 8GB RAM and a 2.4GHz Intel Core i5 processor.\footnote{Mention of commercial products does not imply NIST's endorsement.} We assume that the breach probability can be approximated (in the regime of interest) using a function of the form $q_i(s) = (1 + \kappa_i s)^{-1}$ for all $i \in {\cal A}$. In practice, the parameter $\kappa_i > 0$ models how quickly the breach probability decreases with security investment for system $i$. Here, for simplicity, we take $\kappa_i = \delta_i^{-1}$. In the first example, we use an artificial scale-free network, whereas the second example makes use of an Internet peer-to-peer network.

\textbf{Example 1:} 
We consider a weakly connected network consisting of two MSCCs, denoted by ${\cal G}_1 = ({\cal V}_1, {\cal E}_1)$ and ${\cal G}_2 = ({\cal V}_2, {\cal E}_2)$ with $|\mathcal{V}_1| = 50$ and $|\mathcal{V}_2|=150$. Here, each ${\cal G}_i$ is a bidirectional scale-free network generated with the power law parameter for node degrees set to 1.5, and the minimum and maximum node degrees equal to 2 and $\lceil 3\log |\mathcal{V}_i| \rceil$, respectively, in order to ensure network connectivity with high probability. We also add 10 directed edges chosen uniformly at random (u.a.r) from ${\cal G}_1$ to ${\cal G}_2$. 


We fix $\delta_i = 0.1$ for all $i \in \mathcal{A}$ and $\mathcal{S} = \mathbb{R}^N_{+}$. The infection rates $\beta_{j,i}$ are chosen u.a.r. between $[0.01, 1]$. We choose 
\[
w(\bs) = \1\T\bs \quad \mbox{ and } \quad 
\cb = (\nu\1 +0.2\cb_{\rm rand}) \circ B\T\1,
\]
where the elements of $\cb_{\rm rand}$ are chosen u.a.r in $(0,1)$ and $\nu \ge 0$ is a varying parameter. We select $\cb$ above to reflect an observation that systems that support more neighbors should, on the average, have larger economic costs modeled by $c_i$ (Section~\ref{sec:Formulation}-A). We select u.a.r 10 nodes in $\mathcal{G}_2$ to have positive primary attack rates $\lambda_i = 0.1$. Here, $\mathcal{G}_1$ is not exposed and we consider $\lambda_1^{\epsilon} = \epsilon$ in our perturbed problem.

\begin{figure}[h]
	\centering
	\includegraphics[scale=0.75]{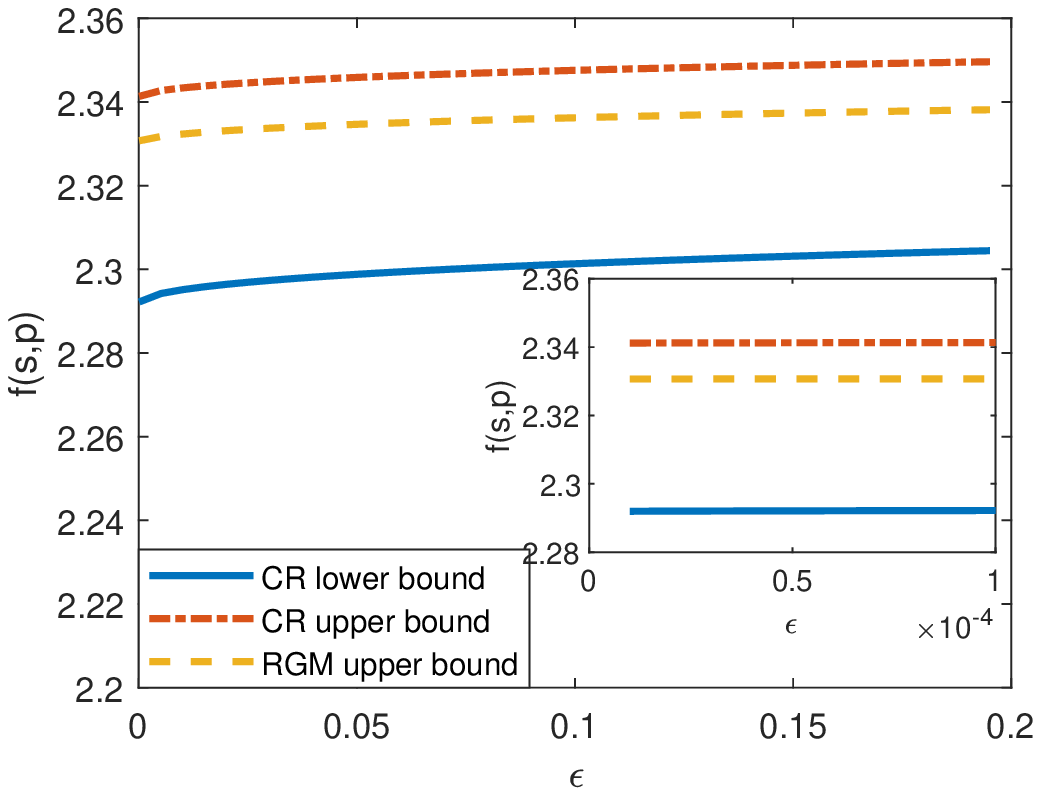}
	\includegraphics[scale=0.75]{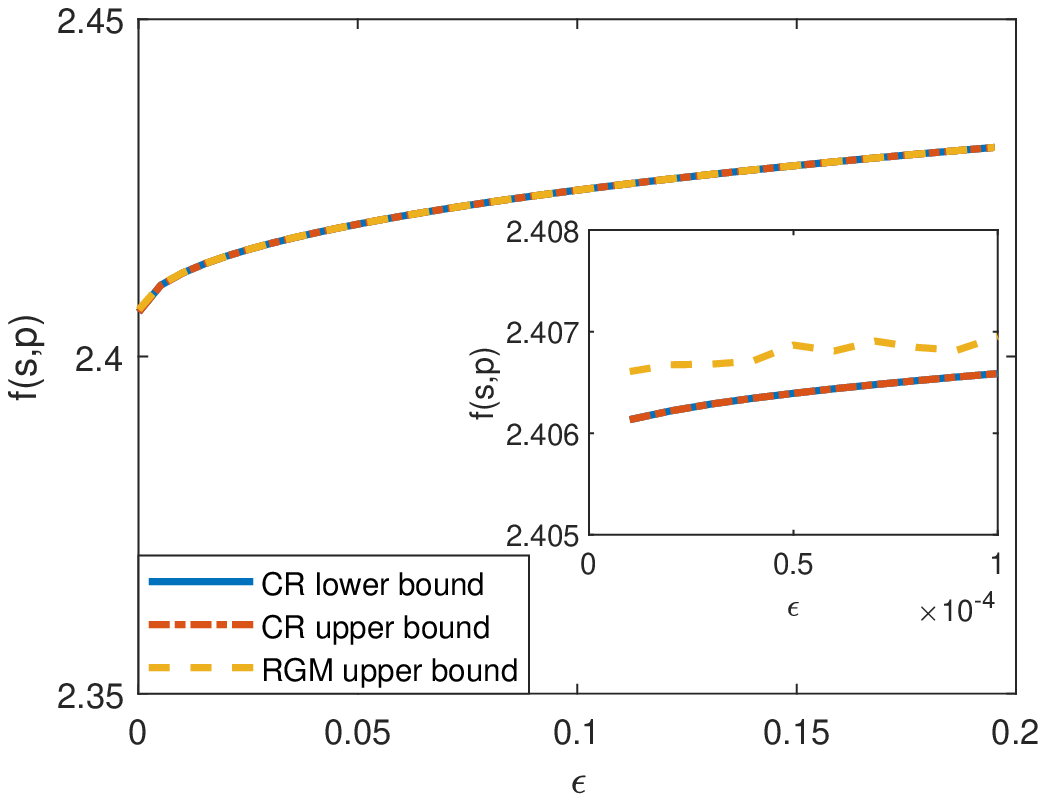}
	\vspace{-3mm}
	\caption{Bounds on optimal costs (normalized by $N$) when varying $\epsilon$ in two scenarios: $\nu=0.9$ in \textit{top plot} and $\nu=1.1$ in \textit{bottom plot}. Insets correspond to a small range of perturbations $\epsilon \in (10^{-5}, 10^{-4})$}
	\label{fig_eg1_eps_range}
	\vspace{-3mm}
\end{figure}

Fig.~\ref{fig_eg1_eps_range} shows the objective function values obtained using (a) the RGM (for finding a local minimizer starting from initial point $(\0, \bar{\pb}(\0))$),\footnote{We combine RGM with a backtracking line search algorithm using  initial step size $\gamma_0 = 1$ and a shrinking factor of 0.85; see \cite{mai_ToN21} for details.} (b) the optimal point $(\bd^+, \pb^+)$ of {\bf [CR]} and (c) the feasible point $(\bd', \pb')$ introduced in Theorem~\ref{thm:1}, as we vary $\epsilon$. Here we solve \textbf{[CR]} using MOSEK package \cite{mosek}.\footnote{Mention of commercial products does not imply NIST's endorsement.} The top plot corresponds to $\nu=0.9$ and shows that {\bf [CR]} is not exact because there is a gap between the objective function values achieved by $(\bd^+, \pb^+)$ and $(\bd', \pb')$. In fact, in this case the RGM provides a better upper bound on the optimal value than $(\bd', \pb')$. The bottom plot corresponds to $\nu=1.1$ and clearly indicates that {\bf [CR]} is exact as both upper and lower bounds overlap. Moreover, we can verify that the condition for {\bf [CR]} to be exact (in Theorem~\ref{thm:1}) holds for any $\nu\ge1$. The insets plot the achieved objective function values over a small range of perturbations $\epsilon \in (10^{-5}, 10^{-4})$. 

For this example, both MOSEK and RGM take less than 1 second to run for most values of $\epsilon$. However, the RGM can take up to 5 seconds when $\nu=1.1$ and the value of perturbation $\epsilon$ is very small (in the interval ($10^{-5}, 10^{-4}$)) because some $\bar{p}_i^{\epsilon}$'s are very close to 0, causing the Jacobian matrix of constraint function~\eqref{eq:g} to become almost singular; see also Lemma~\ref{lemma_M_M_matrix} and its proof. Fig.~\ref{fig_eg1_ps} plots the stable equilibrium $\bar{\pb}(\bs^*_\epsilon)$ and the optimal point $\bs^*_\epsilon$ for the case with $\nu=1.1$ and $\epsilon = 10^{-5}$. Here, the top plot in the figure showing $\bar{\pb}(\bs^*_\epsilon)$ suggests that the first MSCC comprising systems 1 through 50 will likely be fully protected, i.e., achieves disease free stable equilibrium, at the optimal point of the original problem in \eqref{eq:constr1}. In this case, the condition number of the  Jacobian matrix of $g$ in~\eqref{eq:g} is approximately $3.2\times 10^{4}$, explaining longer running times of the RGM as discussed above.
\begin{figure}[h]
	\centering
	\includegraphics[scale=0.75]{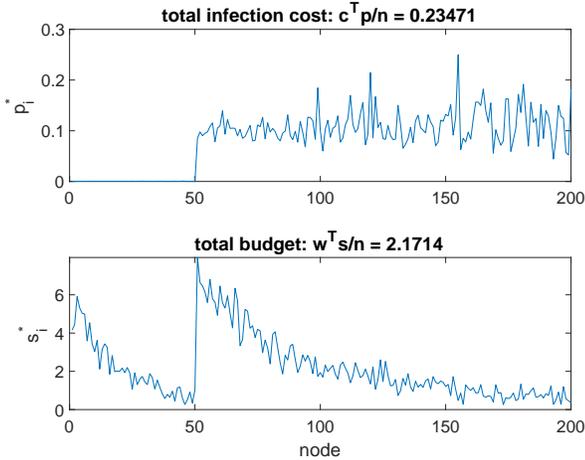}
	\vspace{-3mm}
	\caption{Optimal solution for the case $\nu=1.1$ and $\epsilon = 10^{-5}$.}
	\label{fig_eg1_ps}
\end{figure}

\vspace{3pt}
\textbf{Example 2:} We consider the Gnutella peer-to-peer network from August 9, 2002 with 8,114 nodes and 26,013 directed edges.\footnote{Data is available at  \textsf{https://snap.stanford.edu/data/index.html}.} We use similar settings as in the first example, except that we set $\lambda_i=0.01$ for $i = 1, \ldots, \frac{N}{2}$ and 0 otherwise. In the perturbed problem, we select $\lambda_i^{\epsilon} = \epsilon$ for $i = \frac{N}{2}+1,\ldots,N$. 

\begin{figure}[h]
	\centering
	\includegraphics[scale=0.75]{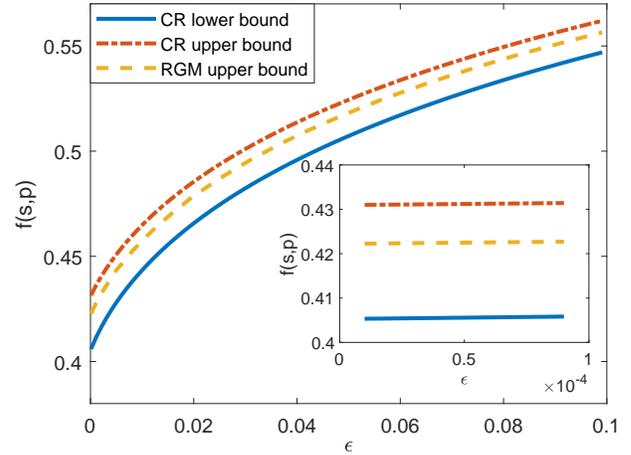}
	\vspace{-3mm}
	\caption{Bounds on optimal costs (normalized by $N$) when varying $\epsilon$ with $\nu=0.8$. Insets correspond to a small range of perturbations $\epsilon \in (10^{-5}, 10^{-4})$}
	\label{fig_eg2_eps_range}
\end{figure}  

Fig.~\ref{fig_eg2_eps_range} shows the lower and upper bounds obtained using the RGM and convex relaxation {\bf [CR]} on the optimal values of the perturbed problem when $\nu=0.8$ and $\cb_{\rm rand}=\0$. In this example, we first use MOSEK to solve \textbf{[CR]}, which takes less than $40$ seconds to run for each value of $\epsilon$. Then, we use the RGM with an initial solution $(\bd',\pb')$ obtained from \textbf{[CR]}, which takes less than $10$ seconds on average. In this example, {\bf [CR]} is not exact as there is a gap between $(\bd^+, \pb^+)$ and $(\bd', \pb')$. However, the gap between that of $(\bd^+, \pb^+)$ and the local minimizer obtained by the RGM is less than a little over 4 percent. Moreover, the inset suggests that all three values of the objective function converge as $\epsilon$ diminishes, validating our analytical results.

\section{Conclusion}\label{sec:Conclusion}

We study the problem of determining (nearly) optimal security investments in large systems, while taking into account temporal dynamics. In contrast to earlier studies, however, we do not assume that the underlying dependence graph among the comprising systems is strongly connected; instead, we allow the dependence graph to be weakly connected. This relaxation of the assumption introduces several technical challenges, in part due to the fact that there could be multiple equilibria. In order to deal with these challenges, we propose a new approach that perturbs the attack rates experienced by the systems, and establish the continuity of the {\em stable} equilibria, which can then be used to design an efficient algorithm.

\bibliographystyle{plain}
\bibliography{ref}

	\end{document}